\documentclass[11pt]{article}
\pdfoutput=1

\usepackage[margin=1in]{geometry}

\DeclareMathAlphabet{\mathbbold}{U}{bbold}{m}{n}
\usepackage{amsmath,amsfonts,amssymb,amsthm}
\usepackage{mathtools}
\usepackage[usenames,dvipsnames,svgnames,table]{xcolor}
\usepackage{thm-restate}

\usepackage[pagebackref]{hyperref}
\hypersetup{
    pdftitle={Quantum distinguishing complexity, zero-error algorithms, and statistical zero knowledge}, 
    pdfauthor={Shalev Ben-David and Robin Kothari}, 
    colorlinks=true, 
    linkcolor=blue, 
    citecolor=blue, 
    urlcolor=blue 
}

\renewcommand{\backref}[1]{}

\renewcommand{\backrefalt}[4]{%
\ifcase #1 %
\or
[p.\ #2]%
\else
[pp.\ #2]%
\fi}

\usepackage{tikz,tikz-qtree}
\usetikzlibrary{backgrounds}
\usepackage{wrapfig}
\usepackage{enumitem}
\usepackage{tocloft}
\usepackage{microtype}

\usepackage[toc]{multitoc}

\makeatletter
\newcommand*\rel@kern[1]{\kern#1\dimexpr\macc@kerna}
\newcommand*\widebar[1]{%
  \begingroup
  \def\mathaccent##1##2{%
    \rel@kern{0.8}%
    \overline{\rel@kern{-0.8}\macc@nucleus\rel@kern{0.2}}%
    \rel@kern{-0.2}%
  }%
  \macc@depth\@ne
  \let\math@bgroup\@empty \let\math@egroup\macc@set@skewchar
  \mathsurround\z@ \frozen@everymath{\mathgroup\macc@group\relax}%
  \macc@set@skewchar\relax
  \let\mathaccentV\macc@nested@a
  \macc@nested@a\relax111{#1}%
  \endgroup
}
\makeatother

\makeatletter
\newcommand{\para}{%
  \@startsection{paragraph}{4}%
  {\z@}{2ex \@plus 3.3ex \@minus .2ex}{-1em}%
  {\normalfont\normalsize\bfseries}%
}
\makeatother

\newtheorem{theorem}{Theorem}
\newtheorem{lemma}[theorem]{Lemma}
\newtheorem{proposition}[theorem]{Proposition}
\newtheorem{corollary}[theorem]{Corollary}
\newtheorem{definition}[theorem]{Definition}

\theoremstyle{definition}

\newcommand{\eq}[1]{\hyperref[eq:#1]{(\ref*{eq:#1})}}
\renewcommand{\sec}[1]{\hyperref[sec:#1]{Section~\ref*{sec:#1}}}
\newcommand{\thm}[1]{\hyperref[thm:#1]{Theorem~\ref*{thm:#1}}}
\newcommand{\lem}[1]{\hyperref[lem:#1]{Lemma~\ref*{lem:#1}}}
\newcommand{\defn}[1]{\hyperref[def:#1]{Definition~\ref*{def:#1}}}
\newcommand{\prop}[1]{\hyperref[prop:#1]{Proposition~\ref*{prop:#1}}}
\newcommand{\cor}[1]{\hyperref[cor:#1]{Corollary~\ref*{cor:#1}}}
\newcommand{\fig}[1]{\hyperref[fig:#1]{Figure~\ref*{fig:#1}}}
\newcommand{\tab}[1]{\hyperref[tab:#1]{Table~\ref*{tab:#1}}}
\newcommand{\alg}[1]{\hyperref[alg:#1]{Algorithm~\ref*{alg:#1}}}
\newcommand{\app}[1]{\hyperref[app:#1]{Appendix~\ref*{app:#1}}}

\newcommand{\<}{\langle}
\renewcommand{\>}{\rangle}
\DeclareMathOperator{\Tr}{Tr}
\DeclareMathOperator{\re}{Re}

\newcommand{\B}{\{0,1\}}
\newcommand{\Ba}{\{0,1,*\}}
\newcommand{\Bo}{\{0,1,\dagger\}}

\newcommand{\id}{\mathbbold{1}}

\newcommand{\Ind}{\textsc{Ind}}
\newcommand{\UInd}{\textsc{UInd}}

\renewcommand{\th}[1]{$#1^\mathrm{th}$}

\DeclareMathOperator{\adeg}{\widetilde{\deg}}

\DeclareMathOperator{\s}{s}
\DeclareMathOperator{\RC}{RC}
\DeclareMathOperator{\QC}{QC}
\DeclareMathOperator{\RS}{RS}

\DeclareMathOperator{\QS}{QD}
\DeclareMathOperator{\QD}{QD}
\DeclareMathOperator{\RD}{RD}

\DeclareMathOperator{\R}{R}
\DeclareMathOperator{\D}{D}
\DeclareMathOperator{\Q}{Q}
\DeclareMathOperator{\C}{C}

\DeclareMathOperator{\Adv}{Adv}

\DeclareMathOperator{\polylog}{polylog}

\DeclareMathOperator{\Dom}{Dom}
\DeclareMathOperator{\QSZK}{QSZK}
\DeclareMathOperator{\SZK}{SZK}

\newcommand{\sab}{\mathrm{sab}}

\newcommand{\tO}{\widetilde{O}}
\newcommand{\tOmega}{\widetilde{\Omega}}
\newcommand{\tTheta}{\widetilde{\Theta}}

\newcommand{\tr}[1]{{\left\lVert#1\right\rVert}_\mathrm{tr}}
\newcommand{\norm}[1]{{\left\lVert#1\right\rVert}}
\newcommand{\cl}[1]{\mathsf{#1}}
\newcommand{\cc}{\mathrm{cc}}
\newcommand{\CC}{\mathrm{cc}}

\begin{document}
\title{\vspace{-1em} Quantum distinguishing complexity,\\ zero-error algorithms, and statistical zero knowledge}

\author{
Shalev Ben-David\\
\small University of Waterloo\\
\small \texttt{shalev.b@uwaterloo.ca}
\and
Robin Kothari \\
\small Microsoft Research\\
\small \texttt{robin.kothari@microsoft.com}
}

\date{}
\maketitle

\begin{abstract}
We define a new query measure we call quantum distinguishing complexity, denoted $\QS(f)$ for a Boolean function $f$.
Unlike a quantum query algorithm, which must output a state close to $|0\rangle$ on a $0$-input and a state close to $|1\rangle$ on a $1$-input, a ``quantum distinguishing algorithm'' can output any state, as long as the output states for any $0$-input and $1$-input are distinguishable. 

Using this measure, we establish a new relationship in query complexity: 
For all total functions $f$, $\Q_0(f)=\tO(\Q(f)^5)$, where $\Q_0(f)$ and $\Q(f)$ denote the zero-error  and bounded-error quantum query complexity of $f$ respectively, improving on the previously known sixth power relationship.

We also define a query measure based on quantum statistical zero-knowledge proofs, $\QSZK(f)$, which is at most $\Q(f)$. 
We show that $\QD(f)$ in fact lower bounds $\QSZK(f)$ and not just $\Q(f)$. 
$\QD(f)$ also upper bounds the (positive-weights) adversary bound, which yields the following relationships for all $f$: $\Q(f) \geq \QSZK(f) \geq \QS(f) = \Omega(\Adv(f)).$ 
This sheds some light on why the adversary bound proves suboptimal bounds for problems like Collision and Set Equality, which have low QSZK complexity.

Lastly, we show implications for lifting theorems in communication complexity. We show that a general lifting theorem for either zero-error quantum query complexity or for QSZK would imply a general lifting theorem for bounded-error quantum query complexity.
\end{abstract}

{\small \tableofcontents}

\clearpage

\section{Introduction}
\label{sec:intro}

In the model of query complexity, we wish to compute some known Boolean function $f:\B^n \to \B$ on an unknown input $x\in\B^n$ that we can access through an oracle that knows $x$. 
In the classical setting, the oracle responds with $x_i$ when queried with an index $i\in[n]$. 
For quantum models, we use essentially the same oracle, but slightly modified to make it unitary. 
The bounded-error quantum query complexity of a function $f$, denoted $\Q(f)$, is the minimum number of queries to the oracle needed to compute the function $f$ with probability greater than $2/3$ on any input $x$.
In other words, the quantum query algorithm outputs a quantum state that is close to $|f(x)\>$.

In this paper we study ``quantum distinguishing complexity,'' a query measure obtained by relaxing the output requirement of quantum query algorithms. 
Essentially, a quantum distinguishing algorithm for $f$ doesn't need to compute $f(x)$, but merely needs to \emph{behave
differently on input $x$ and input $y$ if $f(x)\ne f(y)$}.
We claim that this weaker notion of computation helps shed
light on quantum query complexity and various lower bound
techniques for it. 
We use quantum distinguishing complexity to prove a
new query complexity relationship for total functions:
$\Q_0(f)=O(\Q(f)^5\log\Q(f))$.
We also use it to explain why the non-negative
adversary bound fails for some problems, to provide
lower bound techniques for the query version of the complexity class
$\cl{QSZK}$, and to prove some reductions between lifting
theorems in communication complexity.

\subsection{Quantum distinguishing complexity}

The \emph{quantum distinguishing complexity} of a function $f:D \to \B$ (where $D\subseteq\B^n$), denoted $\QS(f)$, is the minimum number of queries needed to the input $x\in D$ to produce an output state $|\psi_x\>$, such that the output states corresponding to $0$-inputs and $1$-inputs are nearly orthogonal (or far apart in trace distance). 
Note that the usual bounded-error quantum query complexity of a function $f$, denoted $\Q(f)$, is defined similarly with the additional requirement that there should exist a 2-outcome measurement that (with high probability) accepts states corresponding to $1$-inputs and rejects states corresponding to $0$-inputs. 
Since measurements can only distinguish nearly orthogonal states, every quantum algorithm for computing $f$ satisfies the definition of quantum distinguishing complexity. Hence for all functions $f$, we have $\QS(f) \leq \Q(f)$.
We formally define quantum distinguishing complexity and establish some basic properties in \sec{sabotage}.

This is a natural relaxation of bounded-error quantum query complexity and has been mentioned in passing in several prior works. 
Indeed, Barnum, Saks, and Szegedy call this measure $\textrm{DQA}(f)$ in an early technical report~\cite[Remark 1]{BSS01}. 
This measure often comes up in discussions about the (positive-weights) adversary bound,\footnote{The positive-weights adversary bound should not be confused with the stronger negative-weights adversary bound (also known as the general adversary bound), which essentially equals quantum query complexity~\cite{HLS07,LMR+11}.} a lower bound for quantum query complexity introduced by Ambainis~\cite{Amb02}.
The (positive-weights) adversary bound, which we denote by $\Adv(f)$, has several variants~\cite{Amb02,Amb03,BSS03,LM04,Zha05}, which are all essentially the same~\cite{SS06}.
It was noted in several works~\cite{BSS03,HLS07} that the proof that the adversary bound lower bounds quantum query complexity only uses the fact that the outputs corresponding to $0$-inputs and $1$-inputs are nearly orthogonal, and hence for all functions $\QS(f) = \Omega(\Adv(f))$. However, it is not the case that $\QS(f) = \Theta(\Adv(f))$ for all $f$, and we exhibit functions separating these measures. 

Lastly, we show in \sec{sabotage} that this measure is the quantum analogue of a lower bound method for randomized query complexity called randomized sabotage complexity~\cite{BK16}. Hence this measure could also be called ``quantum sabotage complexity.''

\subsection{Fifth power query relation}
Our first result establishes a new relation between query measures for total functions. 
A total function is  a function of the form $f:\B^n \to \B$, as opposed to a partial function, which is a function of the form $f:D\to\B$, where $D \subseteq \B^n$.
We show a new upper bound on the zero-error quantum query complexity of $f$, denoted $\Q_0(f)$, in terms of its quantum distinguishing complexity, and hence its quantum query complexity. 
The zero-error quantum query complexity of $f$ is the minimum number of queries needed by a quantum algorithm that either outputs the correct answer $f(x)$ on input $x$, or outputs \texttt{?} indicating that it does not know, but does this with probability at most $1/2$ on any input $x$.
In \sec{zero} we prove the following.

\begin{restatable}{theorem}{QzQS}
\label{thm:QzQS}
For all total functions $f:\B^n \to \B$, we have 
\begin{equation}
\Q_0(f) = O(\QS(f)^5 \log \QS(f)) = O(\Q(f)^5 \log \Q(f)).    
\end{equation}
Additionally, the algorithm also outputs a certificate for $f(x)$ when it outputs $f(x)$.
\end{restatable}

This is an improvement over the previous best relationship between zero-error and bounded-error quantum query complexity,  $\Q_0(f) = O(\Q(f)^6)$~\cite{BBC+01}, which follows from $\D(f) = O(\Q(f)^6)$, where $\D(f)$ is deterministic query complexity. 
In fact, our result is the first upper bound on zero-error quantum query complexity that does not follow from an upper bound on zero-error randomized query complexity. 
Our proof is borrows ideas from the classical result $\R_0(f) = O(\R(f)^2 \log \R(f))$~\cite{Mid05,KT13}, which is essentially optimal
due to a nearly matching  separation by Ambainis et al.~\cite{ABB+15}.

\subsection{Quantum statistical zero knowledge} 
Next we show that, surprisingly, quantum distinguishing complexity lower bounds a more powerful model of computation than quantum query complexity: the query complexity of computing a function using a quantum statistical zero-knowledge (QSZK) proof system.
A QSZK proof system is an interactive protocol between a quantum verifier and a computationally unbounded, but untrusted prover in which the verifier learns the value of $f(x)$ but learns essentially no more. 
QSZK can also be characterized in terms of its complete problem Quantum State Distinguishability~\cite{Wat02,Wat09}.

In \sec{QSZK}, we discuss the history of quantum statistical zero-knowledge proofs and define an associated query measure $\QSZK(f)$ based on the complete problem Quantum State Distinguishability.
We establish some basic properties of our definition, such as $\QSZK(f) \leq \Q(f)$, which corresponds to the complexity class containment $\cl{BQP} \subseteq \cl{QSZK}$.
We then show that quantum distinguishing complexity lower bounds QSZK complexity.

\begin{restatable}{theorem}{QSQSZK}
\label{thm:QSQSZK}
For all (partial) Boolean functions $f$, $\QS(f) \leq \QSZK(f)$.
\end{restatable}

As a corollary of \thm{QSQSZK} and $\QS(f) = \Omega(\Adv(f))$, we have for all (partial) functions $f$,
\begin{equation}
\Q(f) \geq \QSZK(f) \geq \QS(f) = \Omega(\Adv(f)).    
\end{equation}
This sheds some light on why the adversary bound sometimes proves poor lower bounds: it lower bounds a more powerful model of computation!
For example, it is well known that the adversary bound cannot prove a super-constant lower bound for the collision problem~\cite{AS04}.
It is also easy to see that the collision problem has a constant-query QSZK (and even classical SZK) protocol.

On the bright side, this gives us a new way to prove lower bounds on QSZK query complexity and prove oracle separations against the complexity class $\cl{QSZK}$. 
For example, since we know the OR function on $n$ bits has $\Adv(\textrm{OR}) = \Omega(\sqrt{n})$, this yields an oracle $A$ such that $\cl{NP}^A \nsubseteq \cl{QSZK}^A$, since the OR function has small certificates.
A similar strategy was used recently by Menda and Watrous to show oracle separations against $\cl{QSZK}$~\cite{MW18}.

\subsection{Comparison with other lower bounds}

We compare quantum distinguishing complexity to the two main lower bound techniques for quantum query complexity: the (positive-weights) adversary bound and the polynomial method. Recall that the negative-weights adversary or general adversary completely characterizes quantum query complexity, so we do not compare quantum distinguishing complexity with it.

As noted earlier, the adversary bound is weaker than quantum distinguishing complexity since for all (partial) functions $f$, $\QS(f) = \Omega(\Adv(f))$.
This implies that $\QS(f)$ coincides with $\Q(f)$ for most functions studied in the literature, since most quantum lower bounds are proved using the adversary method. 
Moreover, not only is quantum distinguishing complexity always larger than the adversary bound, it can be exponentially larger for partial functions and quadratically larger for total functions as we show in \thm{QSAdvdeg}. 

Another popular lower bound technique is the polynomial method~\cite{BBC+01}, which uses the fact that the approximate degree of a function lower bounds $\Q(f)$. 
The approximate degree of a Boolean function $f$, denoted $\adeg(f)$, is the minimum degree of a real polynomial $p(x)$ over the input variables such that for all inputs $x$ we have $|f(x)-p(x)|\leq 1/3$. 

\setlength{\intextsep}{0pt}%
\setlength{\columnsep}{20pt}%
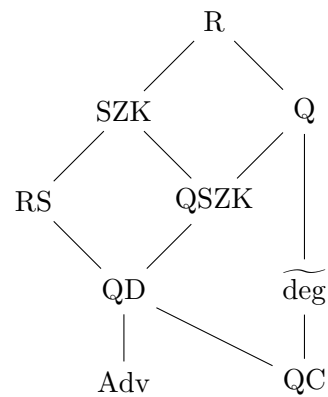
\begin{wrapfigure}{r}{0.30\textwidth}
\centering
   \begin{tikzpicture}[x=1.2cm,y=1.2cm]

     \node (R) at(2,4){$\R$};
     \node (RS) at(0,2){$\RS$};
     \node (QSZK) at(2,2){$\QSZK$};
     \node (SZK) at(1,3){$\SZK$};
     \node (Q) at(3,3){$\Q$};
     \node (adeg) at(3,1.05){$\adeg$};
     \node (QS) at(1,1){$\QS$};
     \node (adv) at(1,0){$\Adv$};
     \node (QC) at(3,0){$\QC$};
     
     \path[-] (SZK) edge (R);
     \path[-] (Q) edge  (R);
     \path[-] (RS) edge (SZK);
     \path[-] (QSZK) edge (SZK);
     \path[-] (QSZK) edge (Q);
     \path[-] (adeg) edge (Q);
     \path[-] (QS) edge (RS);
     \path[-] (QS) edge (QSZK);
     \path[-] (adv) edge (QS);    
     \path[-] (QC) edge (adeg);
     \path[-] (QC) edge (QS);    
   \end{tikzpicture}
    \caption{Relationships between measures. An upward line indicates that a measure is asymptotically upper bounded by the other measure. E.g., for all (partial) functions $f$, $\Q(f) = O(\R(f))$.\label{fig:relations}}
\end{wrapfigure}

We do not know an exponential separation  between quantum distinguishing complexity and approximate degree (for a partial function), since it is not even known if quantum query complexity can be exponentially larger than approximate degree for a partial functions. 
We do, however, show in \thm{QSAdvdeg} that quantum distinguishing complexity can be polynomially larger than approximate degree for total functions.

\begin{restatable}{theorem}{QSAdvdeg}
\label{thm:QSAdvdeg}
There exist total functions $f$ and $g$ with
\begin{equation}
\QS(f) = \tOmega(\Adv(f)^2) \enskip \mathrm{and} \enskip  \QS(g) \geq \adeg(g)^{4-o(1)}.
\end{equation}
There also exists an $n$-bit partial function $h$ with 
\begin{equation}
\QS(h) = \tOmega(n^{1/3}) \enskip  \mathrm{and}  \enskip  \Adv(h) = O(\log n).
\end{equation}
\end{restatable}

This theorem is proved in  \sec{separations}. \fig{relations} shows the known relationships between all the measures discussed in this paper. The measures $\RS$ and $\QC$ are introduced later, and refer to randomized sabotage complexity and quantum certificate complexity, respectively. 

\subsection{Lifting theorems}
Most measures in query complexity have an analogous measure in communication complexity, which we denote with the superscript cc, such as $\Q^\mathrm{\CC}(F)$ and $\QSZK^\mathrm{\CC}(F)$. 
A lifting theorem is a result that transfers a lower bound on a query function $f$ to a lower bound in communication complexity for a lifted version of the function $f$, obtained by composing the function $f$ with a hard communication problem $G$. For example, a lifting theorem is known for deterministic protocols, which means there exists a communication problem $G$ such that for all functions $f$, $\D^\mathrm{\CC}(f \circ G) = \Omega(\D(f))$~\cite{RM99,GPW15}.

Lifting theorems have been shown for some measures, such as nondeterministic query complexity~\cite{GLM+16} and (zero-error or bounded-error) randomized query complexity~\cite{GPW17}, and remain open for measures like zero-error and bounded-error quantum query complexity.
Our next result, proved in \sec{lifting}, shows that if we could prove a lifting theorem for zero-error quantum query complexity or for QSZK query complexity, then we would get a lifting theorem for bounded-error quantum query complexity.

\begin{theorem}[informal]
If a general lifting theorem holds using some gadget $G$ for either zero-error quantum query complexity, i.e., $\Q_0^\mathrm{\CC}(f \circ G) = \tOmega(\Q_0(f))$, or for quantum statistical zero-knowledge protocols, i.e., $\QSZK^\mathrm{\CC}(f \circ G) = \tOmega(\QSZK(f))$, then we obtain a general lifting theorem for bounded-error quantum query complexity (up to logarithmic factors) with the same gadget $G$.
\end{theorem}

In fact, the same conclusion follows from a weaker assumption. We can assume that the lifting theorem proves a lower bound on bounded-error quantum communication complexity assuming a lower bound on quantum distinguishing complexity. In other words, we can assume a lifting theorem of the form $\Q^\mathrm{\CC}(f \circ G) = \tOmega(\QS(f))$, which is weaker than a QSZK lifting theorem since it assumes a stronger lower bound and proves a weaker one.

\section{Preliminaries}
\label{sec:prelim}

We assume the reader is generally familiar with quantum computation~\cite{NC00} and query complexity (for more details, see \cite{BdW02}). We do not assume the reader is familiar with statistical zero-knowledge protocols.

For any positive integer $n$, let $[n] = \{1,\ldots,n\}$. We use $f(n) = \tO(g(n))$ to mean there exists a constant $k$ such that $f(n) = O(g(n) \log^k g(n))$ and similarly $f(n) = \tOmega(g(n))$ means $f(n) = \Omega(g(n)/\log^k g(n))$ for some constant $k$.

\subsection{Distance measures}
For any matrix $A$, we define the spectral norm of $A$, denoted $\norm{A}$ as the largest singular value of $A$. The $1$-norm of $A$, denoted $\norm{A}_1$, is defined as $\Tr\big(\sqrt{A^\dagger A}\big)$, which is also equal to the sum of the singular values of $A$. 

We define the trace distance between two quantum states $\rho$ and $\sigma$ as $\tr{\rho-\sigma} = \frac{1}{2}\norm{\rho-\sigma}_1$. 
The factor of $1/2$ makes this distance measure lie between $0$ and $1$ for density matrices.
Trace distance is a useful distance measure since it exactly captures distinguishability of states and is non-increasing under quantum operations~\cite[Th.~9.2]{NC00}.
For pure states $|\psi\>$ and $|\phi\>$, trace distance is related to their inner product as follows~\cite[eq.~1.186]{Wat18}. 
\begin{equation}
\tr{|\psi\rangle\langle\psi| - |\phi\rangle\langle\phi|} = \sqrt{1-|\<\psi|\phi\>|^2}.\label{eq:tr}
\end{equation}

\subsection{Quantum query complexity} 
In query complexity, we wish to compute a Boolean function $f$ on an input $x$ given query access to the bits of $x$. 
In this paper, we will mostly deal with functions with Boolean input and output.
An $n$-bit function $f:\B^n \to \B$ is called a \emph{total function}. 
An $n$-bit function $f:D\to\B$, where $D \subseteq \B^n$, is called a \emph{partial function} since it is defined on a subset of $\B^n$.
We will also refer to this subset $D$ as the domain of $f$, or $\Dom(f)$.
The goal in query complexity is to compute $f(x)$ while making the fewest queries to the oracle for the bits of $x$.

Classical algorithms have access to an oracle that given an index $i\in [n]$ outputs $x_i$, the \th{i} bit of $x$. 
A quantum algorithm is allowed access to a unitary map that implements this oracle, and is usually taken to be the unitary $O_x$ which acts as follows on inputs $i\in[n]$ and $b\in\B$:
$O_x |i,b\> = |i,b \oplus x_i\>$.
A quantum algorithm that uses the gate $O_x$ in its circuit $k$ times is said to have made $k$ queries to the oracle. 

Since we do not count the complexity of any other gates used in the algorithm, we can assume a $k$-query quantum algorithm always starts with the all-zeros state $|0^m\>$ and applies an oracle-independent unitary $U_0$ followed by the oracle $O_x$ and so on. 
Thus a $k$-query quantum algorithm is specified by $k+1$ oracle-independent unitaries $U_0, \ldots, U_k$, which act on $m$ output qubits. 
The state output by the quantum algorithm is $|\psi_x\> = U_k O_x U_{k-1} O_x \cdots O_x U_1 O_x U_0 |0^m\>$,  where $O_x$ is implicitly $(O_x \otimes \id)$ if $U_i$ acts on more qubits than $O_x$. 
If the quantum algorithm outputs a mixed state, then we assume it traces out some subset $S$ of the $m$ qubits, and hence outputs $\Tr_S(|\psi_x\>\<\psi_x|)$.
If the quantum algorithm outputs a bit, then we assume it measures the first qubit in the standard basis and outputs the result of that measurement.

We can now define the various complexity measures associated with quantum query complexity. 
We say the \emph{bounded-error quantum query complexity} of computing a Boolean function $f$, $\Q(f)$, is the minimum $k$ such that there exists a $k$-query quantum algorithm that on every $x \in \Dom(f)$ outputs $f(x)$ with probability greater than or equal to $2/3$. 
As usual, the constant $2/3$ is unimportant as long as it is a constant strictly greater than half, due to standard error reduction.

A \emph{zero-error quantum algorithm} (or a Las Vegas quantum algorithm) never outputs an incorrect answer on an input $x \in \Dom(f)$, but is allowed to claim ignorance and answer \texttt{?} with probability at most $1/2$. 
The \emph{zero-error quantum query complexity} of $f$, $\Q_0(f)$ is the minimum number of queries needed for a zero-error quantum algorithm to compute $f$. 
Note that $\Q(f) \leq \Q_0(f)$, since a zero-error algorithm can be turned into a bounded-error algorithm by simply outputting a random bit when the zero-error algorithm outputs \texttt{?}.

For zero-error quantum algorithms, there is a subtlety to do with whether or not the algorithm also produces a classical certificate for the input $x$. 
A certificate for $x$ is a subset of bits of $x$, such that the value of $f(x)$ is completely determined by reading these bits alone. 
A classical zero-error algorithm can always be assumed to output such a certificate without loss of generality.
However, this is not known to be true for zero-error quantum algorithms, and zero-error quantum algorithms that also output a certificate when they output a non-\texttt{?} answer are called \emph{self-certifying} algorithms~\cite{BCW+99}. 
All the zero-error quantum algorithms in this paper are self-certifying, which makes our results stronger since we only prove upper bounds on zero-error quantum algorithms.

\section{Quantum distinguishing complexity}
\label{sec:sabotage}

\subsection{Definition}

We now define quantum distinguishing complexity more formally. 
As explained in the introduction, instead of requiring that the quantum algorithm output the value of the function $f(x)$, as in standard quantum query complexity, we only want the quantum algorithm's outputs to be distinguishable (or nearly orthogonal) for $0$-inputs and $1$-inputs.

As an example of how these definitions differ, consider the collision problem. 
In this problem, we are given an input $x\in [n]^n$ and we are promised that if we view $x$ as a function from $[n]\to[n]$, the function is either $1$-to-$1$ or $2$-to-$1$. 
The goal is to distinguish these two cases under the assumption that the input satisfies this promise. 
In this problem, since every $0$-input and $1$-input differ in exactly half the positions $i\in[n]$, our quantum algorithm can simply create the state $|\psi_x\> = \frac{1}{\sqrt{n}}\sum_i |i,x_i\>$ and the states corresponding to $0$-inputs and $1$-inputs will have trace distance $\Omega(1)$.
Thus this problem has quantum distinguishing complexity $O(1)$, but its quantum query complexity is $\Theta(n^{1/3})$~\cite{AS04}.

\begin{definition}[Quantum Distinguishing complexity]
\label{def:QS}
Let $f:D\to \B$, where $D \subseteq \B^n$, be an $n$-bit partial function. $\QS(f)$ is defined as the smallest integer $k$ such that there exists a $k$-query quantum algorithm that on input $x\in D$ outputs a quantum state $\rho_x$ such that 
\begin{equation}
\forall x,y \in D \textrm{ with } f(x) \neq f(y), \quad \tr{\rho_x - \rho_y} \geq 1/6.
\end{equation}
\end{definition}

Note that the definition is robust to minor changes. First, we allow outputting mixed states, although this does not offer any additional power over only outputting pure states. 
The reason is that we can always assume that the quantum algorithm is pure until the final step where some subset of qubits is traced out. 
But if two states are far apart in trace distance after a partial trace, then they were far apart to begin with since trace distance is non-increasing under partial trace.

The constant $1/6$ in \defn{QS} is also arbitrary and any constant in $(0,1)$ would not change the measure by more than a multiplicative constant. 
This is because we can increase the trace distance between the states by outputting multiple copies of the states. We choose the constant $1/6$ purely for aesthetic reasons: This choice ensures that the result in \thm{QSQSZK} has no constant factors.

\subsection{Properties}

We can now establish some basic properties of quantum distinguishing complexity. First, let us formally show that quantum distinguishing complexity lower bounds quantum query complexity. 

\begin{proposition}
\label{prop:QSQ}
For all (partial) Boolean functions $f$, $\QS(f) \leq \Q(f)$.
\end{proposition}

\begin{proof}
Let $\Q(f)=k$ and consider the $k$-query algorithm that witnesses this fact. 
Let $p_x$ be the probability that this $k$-query algorithm, when run on input $x$, outputs $1$ upon measuring the first qubit.
Since the algorithm computes $f$ with bounded error, we know that for all $1$-inputs $x$, $p_x\geq 2/3$, and for all $0$-inputs $y$, $p_y \leq 1/3$.

Now consider the single-qubit state $\rho_x$, which is obtained by taking the final state of this algorithm, tracing out all the qubits except the first one, and then applying a completely dephasing channel to it. 
This state is $\rho_x = \big( \begin{smallmatrix} 1-p_x & 0\\ 0&p_x \end{smallmatrix} \big)$.
Thus for all $x,y$ with $f(x)\neq f(x)$, $\tr{\rho_x - \rho_y} = |p_x - p_y| \geq 1/3$.
\end{proof}

As noted in the introduction, quantum distinguishing complexity is also lower bounded by the adversary bound, i.e., 
\begin{equation}
\QS(f) = \Omega(\Adv(f)).
\end{equation}
We do not prove this since this follows from the arguments that establish that the adversary bound is a lower bound on quantum query complexity~\cite{Amb02,Amb03,BSS03,LM04,Zha05,SS06}, since all these proofs only use the fact that the states output on $0$-inputs and $1$-inputs are nearly orthogonal.

Quantum distinguishing complexity is also superior to quantum certificate complexity $\QC(f)$, as we show in \prop{QSQC}. Quantum certificate complexity is a lower bound on quantum query complexity defined by Aaronson~\cite{Aar08}. It was later shown that quantum certificate complexity also lower bounds approximate polynomial degree~\cite{KT13}. 

Before proving \prop{QSQC}, we first define certificate complexity, randomized certificate complexity, and quantum certificate complexity.

\begin{definition}[Certificate complexity]
\label{def:cert}
For any (partial) function $f$ and input $x\in \Dom(f)$, consider the partial function $f^x$ defined on the domain $\{x\} \cup \{y\in\Dom(f):f(y)\neq f(x)\}$ that satisfies $f^x(x)=1$ and $f^x(y)=0$ for all $y\in\Dom(f)$ with $f(y) \neq f(x)$. 

We define 
the certificate complexity of $f$, denoted $\C(f)$, 
the randomized certificate complexity of $f$, denoted $\RC(f)$, and the quantum certificate complexity of $f$, denoted $\QC(f)$, as follows:
\begin{equation}
\C(f) = \max_{x \in \Dom(f)} \D(f^x), \quad
\RC(f) = \max_{x \in \Dom(f)} \R(f^x), \quad \mathrm{and} \quad
\QC(f) = \max_{x \in \Dom(f)} \Q(f^x).
\end{equation}
\end{definition}

The problem $f^x$ is clearly no harder than computing $f$ itself in any model of computation, and hence these are lower bounds on their respective measures, i.e., $\C(f) \leq \D(f)$, $\RC(f) \leq \R(f)$, and $\QC(f) \leq\Q(f)$. 
We can now prove that $\QS(f)$ is a better lower bound on $\Q(f)$ than $\QC(f)$.
\begin{proposition}
\label{prop:QSQC}
For all (partial) Boolean functions $f$, $\QS(f) = \Omega(\QC(f))$.
\end{proposition}

\begin{proof}
Let $\QS(f) = k$ and consider the $k$-query quantum algorithm that witnesses this fact. 
We can use this  algorithm to solve $f^x$ for any $x\in \Dom(f)$.
Consider the output of the algorithm on input $x$ before the partial trace operation and call this $|\psi_x\>$. 
The trace distance between $|\psi_x\>$ and $|\psi_y\>$ for $y\in \Dom(f)$ with $f(y) \neq f(x)$ is at least $1/6$ since trace distance is non-increasing under partial trace~\cite[Th.~9.2]{NC00}.

Now we construct an algorithm for $f^x$ from this algorithm to show that $\Q(f^x) = O(\QS(f))$. To do so, we run the supposed algorithm and measure whether the output state is $|\psi_x\>$ or not and accept only when the measurement accepts. This yields an algorithm that outputs $1$ on $x$ with probability $1$ and accepts inputs $y$ with $f(x) \neq f(x)$ with some constant probability strictly less than $1$. 
More precisely, the acceptance probability is $|\<\psi_x|\psi_y\>|^2 \leq 1-(1/6)^2$ due to the relationship between inner product and trace distance for pure states. 
Repeating this algorithm a constant number of times yields a bounded-error quantum algorithm for $f^x$.
\end{proof}

\subsection{Relation with randomized sabotage complexity}

We start by reviewing the definition of randomized sabotage complexity,
as presented in \cite{BK16}.
Fix a (partial) Boolean function $f:D\to\B$ with $D\in\B^n$.
For any pair $x,y\in\Dom(f)$ such that $f(x)\ne f(y)$, let $p\in\Ba^n$
be the partial assignment of all bits where $x$ and $y$ agree
(with the symbol $*$ used for the bits where $x$ and $y$ disagree).
We call $p$ a ``sabotaged input'', imagining that a saboteur
replaced bits of $x$ with $*$ symbols until it was no longer
possible to determine $f(x)$.

Let $S_*\subseteq\Ba^n$ be the set of all sabotaged inputs to $f$,
that is, the
set of all partial assignments that are consistent with both a $0$-input
and a $1$-input to $f$. Let $S_\dagger\in\Bo^n$ be the same as $S_*$,
except that the $\dagger$ symbol is used instead of the $*$ symbol.
Finally, let $f_{\sab}:S_*\cup S_\dagger\to\B$ be the function
that takes a sabotaged input and identifies whether it has $*$ symbols
or $\dagger$ symbols, promised that it contains only one type of symbol.
Intuitively, $f_\sab$ is a decision problem
that forces an algorithm computing it to find a $*$ or $\dagger$.
We then define $\RS(f):=\R_0(f_\sab)$, the expected running time
of a zero-error randomized algorithm computing $f_\sab$.

To show that $\RS(f)$ is larger than $\QS(f)$ for all $f$, we
will define a classical measure analogous to $\QS(f)$.
We will then show this measure is equivalent to $\RS(f)$.

\begin{definition}[Randomized distinguishing complexity]
Let $f:D\to\B$, where $D\subseteq\B^n$, be an $n$-bit partial function.
$\RD(f)$ is defined as the smallest integer $k$ such that
there exists a $k$-query randomized algorithm that on input $x\in D$
outputs a sample from a probability distribution $d_x$ such that
\begin{equation}\forall x,y\in D\mbox{ with } f(x)\ne f(y),\quad D_\mathrm{TV}(d_x,d_y)\ge 1/6,\end{equation}
where $D_\mathrm{TV}(\cdot,\cdot)$ stands for the total variation distance between probability distributions.
\end{definition}

Since quantum algorithms can simulate classical algorithms,
we immediately get that $\QS(f)\le\RD(f)$. Next, we will show
that $\RD(f)=\Theta(\RS(f))$, completing the argument that
$\QS(f)=O(\RS(f))$.

\begin{theorem}
Let $f$ be a partial Boolean function. Then
$\RS(f)/12\le\RD(f)\le(12/11)\RS(f)$.
\end{theorem}

\begin{proof}
First, we show that $\RS(f)\le 12\RD(f)$. Let $A$ be an optimal randomized
algorithm for $\RD(f)$, that on input $x$ outputs a sample from
the distribution $d_x$. Let $z\in\Dom(f_\sab)$ be a sabotaged
input, and consider running $A$ on $z$. Since $z$ is sabotaged,
there are inputs $x$ and $y$ with $f(x)\ne f(y)$ that are both
consistent with the non-$*$, non-$\dagger$ bits of $z$.
The variation distance between $d_x$ and $d_y$ is at least $1/6$.

A randomized algorithm can be viewed as a probability distribution
over deterministic algorithms. Split the support of the distribution
for $A$ into two parts:
a set $S$ consisting of deterministic algorithms that,
when run on $z$, query a $*$ or $\dagger$, and a set $T$ consisting
of deterministic algorithms that don't query a $*$ or $\dagger$
when run on $z$. Note that algorithms in $T$ behave the same
on $x$ and $y$. If $A$ samples an algorithm from $T$ with
probability $p$, the total variation distance between the run of $A$
on $x$ and the run of $A$ on $y$ must therefore be at most $2(1-p)$.
Since this is at least $1/6$, we have $p\le 11/12$. Hence
when $A$ is run on $z$, it queries a $*$ or $\dagger$
with probability at least $1/12$.

If we repeat $A$ whenever it does not query a $*$ or $\dagger$,
we get an algorithm that always finds such an entry
and uses at most $12\RD(f)$ queries on expectation. This
is a zero-error randomized algorithm for $f_\sab$,
so $\RS(f)\le 12\RD(f)$.

We now handle the other direction, showing $\RD(f)\le(12/11)\RS(f)$.
Let $A$ be an optimal zero-error randomized algorithm for $f_\sab$.
It makes $\RS(f)$ queries on expectation, and always finds a $*$ or
$\dagger$ in any sabotaged input. Consider the algorithm
$B$ that, on input $x\in\Dom(f)$,
runs $A$ for at most $2\RS(f)$ queries and outputs
the partial assignment it queried (that is, it outputs
all the pairs $(i,x_i)$ that were queried by the algorithm $A$).

Let $x$ and $y$ be inputs to $f$ with $f(x)\ne f(y)$.
Let $z$ be the sabotaged input defined by $x$ and $y$,
that is, $z_i=*$ if $x_i\ne y_i$ and $z_i=x_i=y_i$ otherwise.
By Markov's inequality, after $(12/11)\RS(f)$ queries, $A$ finds
a $*$ with probability at least $1/12$
when it is run on $z$. This means that when $A$ is run on $x$,
it queries an index $i$ for which $x_i\ne y_i$
with probability at least $1/12$. When this happens,
the output of $B(x)$ is not in the support of $d_y$.
This means $d_x$ puts weight at least $1/12$ on symbols not in
the support of $d_y$. Conversely, $d_y$ puts weight at least $1/12$
on symbols not in the support of $d_x$. The total variation distance
between the two distributions is therefore at least $1/6$,
meaning $B$ is a valid $\RD(f)$ algorithm. We conclude
that $\RD(f)\le(12/11)\RS(f)$.
\end{proof}

Combined with $\QS(f)\le\RD(f)$, this theorem gives us
the following corollary.

\begin{corollary}
For all (partial) Boolean functions $f$,
$\QS(f)=O(\RS(f))$.
\end{corollary}

\section{Fifth power query relation}
\label{sec:zero}

In this section we prove a new relationship between zero-error quantum query complexity and quantum distinguishing complexity and bounded-error quantum query complexity, restated below.

\QzQS*

Our proof uses ideas from an analogous classical result~\cite{Mid05,KT13} and the main quantum ingredient used is the hybrid argument of Bennett, Bernstein, Brassard, and  Vazirani~\cite{BBBV97}.
We now describe and prove a version of the hybrid argument that we use.

\subsection{Hybrid argument}

We start by defining the concept of a \emph{sensitive block}. 
For a string $x\in \B^n$ and a subset of input bits $B\subseteq [n]$, which we call a block, we use $x^B$ to denote the input with all bits in $B$ flipped. 
In other words, $x^B$ agrees with $x$ on all positions outside $B$ and disagrees on $B$. For a function $f$ and an input $x\in\Dom(f)$, we say a block $B$ is a sensitive block if $f(x) \neq f(x^B)$.  

Now any algorithm that computes $f$ must also be able to distinguish $x$ from $x^B$, where $B$ is a sensitive block. 
Any algorithm that can distinguish $x$ from $x^B$ must ``look at'' the bits in $B$ in some informal sense.
For classical algorithms, this simply means the algorithm has to query a bit from $B$ with high probability.
The analogous statement for quantum algorithms is not so clear, since quantum algorithms can query all input bits in superposition.
Nevertheless, the hybrid argument still allows us to formalize this intuition in the quantum setting. 
The hybrid argument asserts that the total weight of queries within the sensitive block (i.e., the total sum of probabilities of querying within the sensitive block over the course of the algorithm) cannot be too small~\cite{BBBV97}:

\begin{restatable}[Hybrid Argument]{lemma}{hybrid}
\label{lem:hybrid}
Let $x\in\B^n$ be an input, and let $B\subseteq[n]$
be a block.
Let $Q$ be a $T$-query quantum algorithm
that accepts $x$ and rejects $x^B$ with high probability, or more generally
produces output states that are a constant distance apart in trace
distance for $x$ and $x^B$.

Let $m_i^t$ be the probability that, when $Q$ is run on
$x$ for $t$ queries and then subsequently measured,
it is found to be querying position $i$ of $x$ (i.e., the query register collapses to $|i\>$).
Then
\begin{equation}\sum_{t=1}^T\sum_{i\in B} m_i^t
    =\Omega\left(\frac{1}{T}\right).\end{equation}
\end{restatable}

Note that for a randomized algorithm, we would have $\Omega(1)$ on the right-hand side instead of $\Omega(1/T)$, since a randomized algorithm must look within $B$ (with high probability) at some point during its execution.
This lemma was implicitly proven in \cite{BBBV97}.
We reproduce the proof in \app{hybrid} for the reader's convenience.

\subsection{New upper bound}

To prove our result we also need to upper bound the number of minimal sensitive blocks of a function. It is not too hard to show that any minimal sensitive block has size at most the sensitivity of $f$, $s(f)$, which is the maximum number of sensitive blocks of size $1$ over all inputs $x$. 
Since there are at most $\binom{n}{\s(f)} = O(n^{\s(f)})$ different subsets of $n$ positions of size $\s(f)$, we know that the number of minimal sensitive blocks is at most this quantity. 
Kulkarni and Tal \cite{KT13} improve this simple upper bound replacing $n$ with randomized certificate complexity $\RC(f)$ (\defn{cert}). 

\begin{lemma}\label{lem:KT}
For any total function $f:\B^n\to\B$
and any input $x\in\B^n$, the number of minimal
sensitive blocks of $x$ with respect to $f$ is at most
$O(\RC(f)^{\s(f)})$.
\end{lemma}

We are now ready to prove \thm{QzQS}.

\begin{proof}[Proof of \protect{\thm{QzQS}}.]
Let $Q$ be the optimal quantum distinguishing algorithm for $f$,
that uses $T=\QS(f)$ queries.
Consider running the following quantum algorithm $P$ on oracle input $x\in \B^n$:

\begin{enumerate}
\item Pick $t\in[T]$ uniformly at random.
\item Run $Q$ on $x$ for $t$ queries and measure the query register.
\item Write down (on a classical tape) the position $i$ where $Q$ is found
to be querying, as well as the query output $x_i$.
\end{enumerate}

The algorithm $P$ uses $t \leq T$ quantum queries.
Now that the probability $P$ wrote down the index $i$
is $(1/T)\sum_{t=1}^T m_i^t$. For any block
$B\subseteq[n]$, the probability that $P$ wrote down
some index in $B$ is 
\begin{equation}\frac{1}{T}\sum_{t=1}^T\sum_{i\in B}m_i^t.\end{equation}
If $B$ is a sensitive block for the input $x$,
then the hybrid argument (\lem{hybrid}) implies the probability that our new algorithm
$P$ outputs an index in $B$ is $\Omega(1/T^2)$.

Next, we repeat the algorithm $P$ several times.
We claim that after $O(T^2\s(f)\log\RC(f))$
repetitions, the outputs of $P$ constitute a certificate
for $x$ with constant probability.

To see this, note that for any minimal sensitive block $B$ of the input $x$,
the probability that some run of $P$ (out of the $O(T^2\s(f)\log\RC(f))$ many runs) queries
in the block $B$ is $1-O(\RC(f)^{-\s(f)})$. 
This is because $T^2$ repetitions boost the probability of querying in a minimal sensitive block from $\Omega(1/T^2)$ to $\Omega(1)$, and then $\s(f)\log\RC(f)$ repetitions of this boosted algorithm further boost the probability to the claimed bound.
Hence, by \lem{KT} and the union bound,
there is a constant
probability that these runs of $P$ query a bit
in every minimal sensitive block of the input $x$.
But a set of bits that intersects every sensitive
block of $x$ is a certificate for $x$.
Thus these runs of $P$ output a certificate for the
input $x$ with constant probability.

Any algorithm that finds a certificate with constant
probability can be turned into a zero-error algorithm
by repeating whenever a certificate is not found.
We therefore get a zero-error algorithm that works
simply by repeating $P$ a sufficient number of times.
Note that $P$ uses $O(T)$ quantum queries and
must be repeated $O(T^2\s(f)\log\RC(f))$ times.
Recalling that $T=\QS(f)$, we get
\begin{equation}\Q_0(f)=O(\QS(f)^3\s(f)\log\RC(f)).\end{equation}
We can simplify this to $\Q_0(f)=O(\QS(f)^5\log\QS(f))$,
since $\s(f) = O(\RC(f)) = O(\QC(f)^2)$~\cite{Aar08} and $\QC(f) = O(\QS(f))$ (\prop{QSQC}).
\end{proof}

\section{Quantum statistical zero knowledge}
\label{sec:QSZK}

\subsection{History}

The subject of statistical zero-knowledge proof systems has a rich history in the classical setting, and the interested reader is referred to the paper of Sahai and Vadhan~\cite{SV03}.
Informally, the complexity class $\cl{SZK}$ contains problems that can be solved by a probabilistic polynomial-time verifier interacting with a computationally unbounded prover (like the class $\cl{IP}$) with the additional restriction that the verifier not learn anything from the prover (statistically) other than the answer to the problem. 
From this it is clear that $\mathsf{BPP} \subseteq \mathsf{SZK}$, since the verifier can simply not interact with the prover, and $\mathsf{SZK} \subseteq \cl{IP}$, since $\cl{IP}$ is simply $\cl{SZK}$ without the zero-knowledge constraint. 

More surprisingly, it is also known that $\cl{SZK} = \cl{coSZK}$, and that we can assume without loss of generality that the interaction is only one round and uses public randomness, which means $\cl{SZK} \subseteq \cl{AM} \cap \cl{coAM}$. 
Another interesting subtlety is that $\cl{SZK}$ can be defined assuming an honest verifier, one who does not deviate from the protocol to learn more, or a cheating verifier, who may deviate from the protocol. 
It turns out that these definitions lead to the same complexity class~\cite{GSV98}.
The class $\cl{SZK}$ also has a much simpler characterization in terms of a complete problem called \emph{statistical difference}, as shown by Sahai and Vadhan~\cite{SV03}, which yields easier proofs of some of these facts. 
Informally, in the statistical difference problem we are given two circuits that sample from probability distributions, and the task is determine whether the distributions are far or close in total variation distance.

On the quantum side, (honest-verifier) $\cl{QSZK}$ was first defined by Watrous~\cite{Wat02}, and like the classical case, it satisfies $\cl{BQP} \subseteq \cl{QSZK} \subseteq \cl{QIP}$.
The same paper strengthened these obvious containments by showing that $\cl{QSZK}$ is closed under complement (i.e., $\cl{QSZK} = \cl{coQSZK}$) and that the protocol can be assumed to be one round, which gives $\cl{QSZK} \subseteq \cl{QIP(2)}$.
Watrous also showed that $\cl{QSZK}$ has a complete problem, called \emph{quantum state distinguishability}, which is a quantum generalization of the statistical difference problem of Sahai and Vadhan. 
In this problem, we are given two quantum circuits outputting mixed states and have to decide if the states are far apart or close in trace distance.
Later, Watrous~\cite{Wat09} also showed that honest-verifier $\cl{QSZK}$ and cheating-verifier $\cl{QSZK}$ are the same, as in the classical case.

\subsection{Definition}

We now define a query  analogue of quantum statistical zero-knowledge. 
Instead of defining $\QSZK(f)$ in terms of an interactive zero-knowledge protocol for $f$, we use the complete problem characterization by Watrous. 
This yields a considerably simpler definition of $\cl{QSZK}$ in the query setting.\footnote{The complete problem is often used to define SZK (and its variants, like NISZK) in query complexity and communication complexity (for example, see \cite{BCHTV17,Sub17}). It is not obvious whether the definition via an interactive proof and the definition via the complete problem coincide exactly as the problem is complete under polynomial-time reductions, which may add polynomial overhead.}

\begin{definition}[QSZK]
\label{def:QSZK}
Let $f:D\to \B$, where $D \subseteq \B^n$, be an $n$-bit partial function. $\QSZK(f)$ is defined as the smallest integer $k$ such that there exists two quantum query algorithms making $k$ queries in total that on input $x\in D$ output states $\rho_x$ and $\sigma_x$ of the same size such that 
\begin{itemize}
    \item $\forall x\in D$ with $f(x)=1, \enspace \tr{\rho_x - \sigma_x} \geq 2/3$,
    \item $\forall x\in D$ with $f(x)=0, \enspace \tr{\rho_x - \sigma_x} \leq 1/3$. 
\end{itemize}

\end{definition}

This definition is also robust to some changes. In particular, the constants $2/3$ and $1/3$ can be replaced by any constants $\alpha \in [0,1]$ and $\beta \in [0,1]$ as long as $\alpha^2 > \beta$.
Hence an alternate definition with $0.99$ instead of $2/3$ and $0.01$ instead of $1/3$ leads to the same complexity measure up to multiplicative constants.
This follows from the analogous property of the complexity class $\cl{QSZK}$, which was shown by Watrous~\cite{Wat02} (see Theorem 1 in the conference version or Theorem 5 in the full version for more details).

\subsection{Properties}

As a sanity check, let us prove the query analog of the obvious containment $\cl{BQP} \subseteq \cl{QSZK}$.

\begin{proposition}
\label{prop:QSZKQ}
For all (partial) Boolean functions $f$, $\QSZK(f) \leq \Q(f)$.
\end{proposition}

\begin{proof}
Let $\Q(f)=k$ and consider the $k$-query algorithm that witnesses this fact. 
Let $p_x$ be the probability that this $k$-query algorithm when run on input $x$ outputs $1$ upon measuring the first qubit.
Since the algorithm computes $f$ with bounded error, we know that $p_x\geq 2/3$ for $1$-inputs and $p_x \leq 1/3$ for $0$-inputs.

Now consider the single-qubit state $\rho_x$, which is obtained by taking the final state of this algorithm, tracing out all the qubits except the first one, and then applying a completely dephasing channel to it. 
This is equivalent to measuring the first qubit in the standard basis and outputting $|b\>$ when the result is $b$.
This state is $\rho_x = \big( \begin{smallmatrix} 1-p_x & 0\\ 0&p_x \end{smallmatrix} \big)$.
Let us also define $\sigma_x$ as $\big( \begin{smallmatrix} 1 & 0\\ 0&0 \end{smallmatrix} \big)$ for all $x$.

Now let us check that the conditions of \defn{QSZK} are satisfied by these states.
For all inputs $x$, we have $\tr{\rho_x - \sigma_x} = |p_x|$.
And we know that $p_x\geq 2/3$ for $1$-inputs and $0\leq p_x \leq 1/3$ for $0$-inputs, which completes the proof.
\end{proof}

The measure $\QSZK(f)$ also satisfies another useful property, that $\QSZK(f)=\Theta(\QSZK(\neg f))$. This is the analogue of the result that $\cl{QSZK} = \cl{coQSZK}$~\cite{Wat02}. Since we do not use this property, we only provide a sketch of the proof.

\para{Sketch of proof of $\QSZK(f)=\Theta(\QSZK(\neg f))$.}
To prove this, we would like to reduce the complete problem to its complement. In other words, we are given two circuits that query an oracle preparing $\rho_x$ and $\sigma_x$ that are either far apart in trace distance (when $f(x)=1$) or close in trace distance (when $f(x)=0$).
From these circuits, we want to define two new states $\rho'_x$ and $\sigma'_x$, such that these states are far when $\rho_x$ and $\sigma_x$ were close, and close when $\rho_x$ and $\sigma_x$ were far. 
Before starting the transformation, we first boost the parameters $2/3$ and $1/3$ to be extremely close to $1$ and $0$ respectively. 
For this sketch we will assume the parameters are exactly $1$ and $0$, which means when the states are far, they are perfectly distinguishable (i.e., $\tr{\rho_x-\sigma_x} = 1$), and when they are close, they are equal (i.e., $\rho_x = \sigma_x)$.

To perform this transformation, consider the pure states output by the circuits before tracing out any qubits. Let $|R_x\>_{BC}$ and $|S_x\>_{BC}$ be the pure state on registers $B$ and $C$, which yields $\rho_x$ and $\sigma_x$, respectively when register $B$ is traced out. More formally, we have
\begin{equation}
\rho_x = \Tr_B(|R_x\>\<R_x|_{BC}) \textrm{ and } \sigma_x = \Tr_B(|S_x\>\<S_x|_{BC}).
\end{equation}
From the pure states $|R_x\>_{BC}$ and $|S_x\>_{BC}$, we define two new pure states on registers $A$, $B$, $C$, and $D$, as follows:
\begin{align}
|R'_x\> &= \frac{1}{\sqrt{2}}\Bigl(|0\>_A|R_x\>_{BC}|0\>_D + |1\>_A|S_x\>_{BC}|0\>_D \Bigr) \textrm{ and }\\ 
|S'_x\> &= \frac{1}{\sqrt{2}}\Bigl(|0\>_A|R_x\>_{BC}|0\>_D + |1\>_A|S_x\>_{BC}|1\>_D \Bigr).
\end{align}
Note that the only difference between these states is in register $D$. If we have circuits preparing states $|R_x\>_{BC}$ and $|S_x\>_{BC}$, it is easy to see that we can construct circuits preparing $|R'_x\>_{ABCD}$ and $|S'_x\>_{ABCD}$.
We now define the states $\rho'_x$ and $\sigma'_x$ from these states by tracing out registers $C$ and $D$:
\begin{equation}
\rho'_x = \Tr_{CD}(|R'_x\>\<R'_x|_{ABCD}) \textrm{ and } \sigma'_x = \Tr_{CD}(|S'_x\>\<S'_x|_{ABCD}).
\end{equation}

We claim that these states satisfy the conditions we require. When $f(x)=1$, we have that $\tr{\rho_x- \sigma_x} = 1$, i.e., the residual state on register $C$ for states $|R'_x\>$ and $|S'_x\>$ is completely distinguishable. In this case, before we trace out registers $C$ and $D$, we could implement a unitary on these registers which reads register $C$ and writes onto register $D$ whether the state in $C$ is $\rho_x$ or $\sigma_x$. This operation maps the state $|R'_x\>$ to the state $|S'_x\>$ and only acts on the traced out qubits, which does not affect the qubits that are not traced out, and we have $\rho'_x = \sigma'_x$.

When $f(x)=0$, we have that $\rho_x = \sigma_x$. In this case we want to show that $\rho'_x$ and $\sigma'_x$ are distinguishable. We will show that after applying a specific unitary to these states are tracing out register $B$, in the first case we are left with the state $|+\>\<+|_A$, but in the second case we have $\frac{1}{2}\id_A$, which can be distinguished.

Since $\rho_x = \sigma_x$, there is a unitary $U_B$ such that $(U_B \otimes \id_C)|R_x\>_{BC} = |S_x\>_{BC}$.
Controlled on the qubit in register $A$, let us apply the unitary $U_B$ to register $B$ of $|R'_x\>$ and $|S'_x\>$ before we trace out registers $C$ and $D$, which is equivalent to applying it after tracing out the registers. 
This makes registers $BC$ unentangled with the rest of the state, and equal to $|S_x\>_{BC}$. In the first case we are left with the state $|+\>_A |0\>_D$ on registers $A$ and $D$, while in the second case we have $\frac{1}{2}(|00\>_{AD} + |11\>_{AD})$. 
Tracing out register $D$ leaves us with the $|+\>$ state in the first case and the maximally mixed state in the second case, as claimed.

\subsection{Relation with adversary bound}

We have already showed that $\QS(f) \leq \Q(f)$ (\prop{QSQ}) and $\QSZK(f) \leq \Q(f)$ (\prop{QSZKQ}). 
We now show that $\QS(f)$ is actually smaller than $\QSZK(f)$. 

\QSQSZK*

\begin{proof}
Let $\QSZK(f) = k$ and consider the quantum algorithms that witnesses this fact. We claim that the tensor product of outputs of these algorithms already satisfies the conditions in \defn{QS} and hence proves $\QS(f) \leq k$.

To see this, observe that the algorithm outputs the state $\rho_x \otimes \sigma_x$ on input $x$, which satisfies the conditions of \defn{QSZK}. More precisely, this means for any $x$ and $y$ such that $f(x)=1$ and $f(y)=0$, we know that $\tr{\rho_x - \sigma_x} \geq 2/3$ and $\tr{\rho_y - \sigma_y} \leq 1/3$. 
We want to show that 
\begin{equation}
\tr{\rho_x \otimes \sigma_x - \rho_y \otimes \sigma_y} \geq 1/6.
\end{equation}
Since trace distance is non-increasing under partial trace, 
we have 
$\tr{\rho_x \otimes \sigma_x - \rho_y \otimes \sigma_y} \geq \tr{\rho_x - \rho_y}$ and 
$\tr{\rho_x \otimes \sigma_x - \rho_y \otimes \sigma_y} \geq \tr{\sigma_x - \sigma_y}$, which imply
\begin{align*} 
&\tr{\rho_x \otimes \sigma_x - \rho_y \otimes \sigma_y} \geq \max\left\{\tr{\rho_x - \rho_y}, \tr{\sigma_x - \sigma_y} \right\}.
\end{align*}
Now if we can show the right-hand side is at least $1/6$, then we are done. To show this, toward a contradiction assume that $\max\left\{\tr{\rho_x - \rho_y}, \tr{\sigma_x - \sigma_y} \right\} < 1/6$. Then we have 
\begin{align*}
    \tr{\rho_x - \sigma_x} &= \tr{\rho_x - \rho_y + \rho_y -\sigma_y + \sigma_y - \sigma_x}\\
    &\leq \tr{\rho_x - \rho_y} + \tr{\rho_y -\sigma_y} + \tr{\sigma_y - \sigma_x}\\
    & < 1/6 + 1/3 + 1/6 = 2/3,
\end{align*}
which contradicts $\tr{\rho_x - \sigma_x} \geq 2/3$.
\end{proof}

As noted, as a corollary of this theorem and $\QS(f) = \Omega(\Adv(f))$, we have for all (partial) functions $f$,
\begin{equation}
\QSZK(f) = \Omega(\Adv(f)).    
\end{equation}
This can be used to prove lower bounds on QSZK protocols for functions. 
For example, consider the OR function and let us try to compute it with an interactive protocol without the zero-knowledge requirement.
It is easy to see that when $\textrm{OR}(x)=1$, a computationally unbounded prover can simply send over the location of a bit $i$ such that $x_i=1$, which can be checked using only $1$ query. 
Of course, this protocol leaks information and in particular lets the verifier know the location of a $1$.
But is it necessary that an efficient protocol for OR must leak information?
Our lower bound says this must be the case, because $\Adv(\mathrm{OR}) = \Omega(\sqrt{n})$ and hence any zero-knowledge protocol for the function must make $\Omega(\sqrt{n})$ queries.

\section{Comparison with other lower bounds}
\label{sec:separations}

In this section, we establish the separations between quantum distinguishing complexity and the adversary bound and the polynomial method claimed in \thm{QSAdvdeg}.

To prove this, we will compose known functions with the index function and establish the behavior of quantum distinguishing complexity under composition with the index function.
This kind of composition was also studied by Chen~\cite{Che16}, who used it to show an oracle separation between $\cl{P}^\cl{SZK}$ and $\cl{QSZK}$.

\subsection{Index functions}

Let $\Ind_k:\B^{k+2^k}\to\B$
denote the index function, the function
that on input $(x,y)$ with $x\in\B^k$
and $y\in \B^{2^k}$, outputs the bit of $y$
indexed by the string $x$.
We wish to study the composition
of the index function with an arbitrary
Boolean function $f$, but composed only
on the first $k$ bits of the index function.
We'll denote this composition by
$\Ind_k\circ_k f$. More precisely, if $f$ is an $n$-bit function,
$\Ind_k\circ_k f$ is a function on $nk+2^k$ bits that evaluates
$f$ on the first $k$ $n$-bit strings to obtain a
binary string $x$ of length $k$, and then uses $x$ to index
into the next $2^k$ bits of the input and outputs the 
bit indexed by $x$.

In addition to the index function, which is total,
we will also study a function we call
the ``unambiguous index function,'' $\UInd_k$.
This is a partial function defined similarly
to the index function, except that
the location of the array $y$ pointed to by the first part of the input is
``marked,'' and we are promised that no other bits of the array are ``marked.''
More explicitly, the function
is defined on $k+2\cdot 2^k$ bits,
with the first $k$ bits indexing a pair
of adjacent bits in the remainder of the input.
So if the first part of the input represents the integer $x$, 
that means it points to the cells $2x$ and $2x+1$ in the second part of the input.
The output of $\UInd_k$ is the first bit of the pair pointed to, i.e., it will be the
bit stored at array location $2x$.
Moreover, we are promised that the second bit of this pair (the bit at array location $2x+1$)
will always be $1$, 
and also that the second bit in every \emph{other} pair (i.e., other than the pair $2x$, $2x+1$) will always be $0$.

Intuitively, there is only one strategy to solve $\Ind_k$, which is to 
read the first $k$ bits and find the cell pointed to.
But to solve $\UInd_k$, there are two good strategies: 
either read the first $k$ bits (and determine $x$), 
or search the remainder of the input
for the unique position where the second bit 
of a pair is $1$, which marks the cell pointed to by $x$.

\subsection{Index function composition}
We now examine the behavior of quantum distinguishing complexity under composition with the Index and Unambiguous Index functions. To prove our result, we need the following strong direct product theorem for quantum query complexity due to Lee and Roland~\cite{LR13}:

\begin{theorem}[Strong direct product]
\label{thm:direct_product}
Let $f$ be a partial Boolean function
with $\Dom(f)\subseteq \B^n$,
and let $f^{(k)}:\Dom(f)^k\to\B^k$
be the task of solving $k$ independent inputs to $f$
simultaneously. Then any quantum
algorithm that solves $f^{(k)}$ with
success probability at least $(5/6)^k$
uses $\Omega(k\Q(f))$ queries.
\end{theorem}

We can now prove our composition theorems.

\begin{theorem}
\label{thm:Ind}
There is a constant $c$ such that
for any partial function $f$, if
$k\geq c\log \Q(f)$, then
\begin{equation}\QS(\Ind_k \circ_k f) = \Theta(\Q(\Ind_k\circ_k f))= \Theta(k\Q(f))\end{equation}
\begin{equation}\QS(\UInd_k \circ_k f) = \Theta(\Q(\UInd_k\circ_k f))= \Theta(k\Q(f)).\end{equation}
In other words, composing a function with a large
enough index gadget makes $\QS$ and $\Q$ coincide.
\end{theorem}

\begin{proof}
Recall that quantum query complexity
composes perfectly \cite{LMR+11},
so $\Q(\Ind_k\circ f)=\Theta(\Q(\Ind_k)\Q(f))
=O(k\Q(f))$. We argue that
$\Q(\Ind_k\circ_k f)$ is smaller than
$\Q(\Ind_k\circ f)$. This is 
because we can convert any algorithm for
$\Q(\Ind_k\circ f)$ into an algorithm
for $\Q(\Ind_k\circ_k f)$: fix a $0$-input
$x^0$ and a $1$-input $x^1$ for $f$;
then, given an input to $\Q(\Ind_k\circ_k f)$,
pretend that each $0$ bit in the second half
of the input is actually $x^0$, and
that each $1$ bit is actually $x^1$
(the algorithm can do this by applying
the appropriate unitary). This
converts the input into an input for
$\Q(\Ind_k\circ f)$, completing the reduction.

Thus $\Q(\Ind_k\circ_k f)=O(k\Q(f))$.
Similarly, $\Q(\UInd_k\circ_k f)=O(k\Q(f))$.
Since $\QS$ is smaller than $\Q$,
it remains only to show that
$\QS(\Ind_k \circ_k f)=\Omega(k\Q(f))$
and $\QS(\UInd_k \circ_k f)=\Omega(k\Q(f))$.
We complete the argument for $\UInd$; the argument
for $\Ind$ is similar.

Let $Q$ be an optimal quantum distinguishing algorithm for
$\UInd_k\circ_k f$. We turn $Q$ into a quantum algorithm
$Q'$ that uses the same number of queries,
and solves all $k$ copies of $f$ with non-negligible
probability; we then apply the direct product
theorem (\thm{direct_product}) to lower bound the number of queries
required by $Q'$, and hence by $Q$.

Given $k$ inputs to $f$, the first thing the algorithm
$Q'$ does is append an all-$0$ array to turn it into
an input to $\UInd_k\circ_k f$. 
(Since the array is all zeros, the new input does not
satisfy the promise of $\UInd_k\circ_k f$,
but we will still be able to run $Q$ on it.)
Then $Q'$ picks a random
number $t$ between $1$ and $T$ uniformly,
where $T=\QS(\UInd_k\circ_k f)$
is the number of queries used by $Q$, and
simulates $Q$ for $t$ queries. The algorithm $Q'$ then
measures the state of $Q$ to determine the position
at which $Q$ was going to query. If this position
is in the array part of the input and is inside
a pair that has index $i\in\B^k$, the algorithm $Q'$
will then output the string $i$.

Consider the correct pair in the array (the
one really pointed to by the $k$ copies of $f$).
Flipping the pair from $00$ to $01$ causes the input to satisfy
the promise of $\UInd_k\circ_k f$, and causes the output to
become a $0$-input. On the other hand,
flipping the pair from $00$ to $11$ causes
the input to become a $1$-input. Let $|\psi\rangle$
be the final state of $Q$ when run on the original, illegal input.
Let $|\psi_0\rangle$ be the final state of $Q$ when run on
the flipped $0$-input, and let $|\psi_1\rangle$
be the final state of $Q$ when run on the $1$-input.
We know that $|\psi_0\rangle$ and $|\psi_1\rangle$
are far in trace distance. Hence $|\psi\rangle$
must be a far in trace distance from at least one on them.

Thus
by \lem{hybrid}, the probability that $Q'$ finds $Q$
querying inside the correct pair of the array is $\Omega(1/T^2)$.
This means that $Q'$ outputs the correct
string of answers to the $k$ inputs to $f$
is with probability at least $\Omega(1/T^2)$.
Since $Q'$ uses only $T$ queries, by \thm{direct_product} we must have
either $T=\Omega(k\Q(f))$ or
$1/T^2=O((5/6)^k)$. The latter implies
$T=\Omega((6/5)^{k/2})=
\Omega((6/5)^{k/4}\cdot (6/5)^{k/4})
=2^{\Omega(k)}\cdot 2^{\Omega(k)}$.
When $k\geq c\log\Q(f)$ for a large enough
constant $c$, this gives
$T\geq 2^{\Omega(k)}\Q(f)=\Omega(k\Q(f))$.
Recalling that $T=\QS(\UInd_k\circ_k f)$, we get
$\QS(\Ind_k\circ_k f)=\Omega(k\Q(f))$, as desired.
\end{proof}

\subsection{Separations}
Using this theorem we can now establish \thm{QSAdvdeg}, restated for convenience:

\QSAdvdeg*

\begin{proof}
There exists an $n$-bit total function $f'$ with a quadratic separation between quantum query complexity and the adversary bound, i.e., 
$\Q(f')=\tOmega(\Adv(f')^2)$. The function is $k$-sum with $k\approx \log n$ (see \cite{BS13,ABK16} for more details). 
Now consider the function $f = \Ind_k \circ f'$, where $k=\Omega(\log\Q(f))$.
By \thm{Ind}, the $\QS$ of these functions increases to $\Q$.
However, since the adversary bound satisfies a composition theorem~\cite{HLS07}, its value only increases by a factor of $k$.
Thus $\QS(f) = \tOmega(\Adv(f)^2)$.

Similarly, if we start with the collision problem which has $\Q(h') = \Theta(n^{1/3})$~\cite{AS04}, but $\Adv(h')=O(1)$, and define $h = \Ind_k \circ h'$ for $k=\Theta(\log n)$, then $\QS(h) = \tOmega(n^{1/3})$ but $\Adv(f) = O(\log n)$.

There also exist total functions with
$\Q(g')\geq \adeg(g')^{4-o(1)}$~\cite{ABK16}.
Composing this function with $\Ind_k$ on the first $k$ bits
with $k=\Omega(\log\Q(f))$ yields a function $g$ with the desired separation, 
since approximate polynomial degree also composes in the upper bound direction~\cite{She12}.
\end{proof}

\section{Lifting theorems}
\label{sec:lifting}

\subsection{Background}
Lifting theorems are results that relate communication complexity measures to query complexity
measures. 
For a fixed query measure, such as $\D(f)$,
and a communication complexity measure that
intuitively corresponds to it, such as
deterministic communication complexity $\D^{\CC}(F)$,
we may hope to be able to prove a theorem of the form:
there exists some communication gadget $G$ such that
$\D^{\CC}(f\circ G)=\tTheta(\D(f))$.
In fact, when the size of the gadget $G$ is allowed
to depend on the input size of $f$ and the
$\tTheta$ is allowed to hide $\polylog n$ factors,
such a result is known~\cite{GPW15}.

We remark that the upper bound direction---showing
the communication measure of $f\circ G$ is at most
the corresponding query measure of $f$---is usually
easy. We can simulate the query algorithm in
the communication complexity world, losing
only a multiplicative factor that depends on
the difficulty of computing $G$. The lower bound
direction, which lower bounds a communication
complexity measure by a query complexity measure,
is usually much harder, and is what we will usually
refer to when we use the term ``lifting theorem.''

The result of \cite{GPW15} gives a lifting theorem
for deterministic protocols, which we will denote 
by $\D\to\D^{\CC}$ to mean it transfers a lower bound on the first
measure to a lower bound on the second. 
Recently, lifting theorems have been shown for 
$\R$ and $\R_0$ (with the corresponding communication complexity
measures being the obvious ones: randomized
communication with bounded error and randomized
communication with zero error, denoted $\R^{\CC}$
and $\R_0^{\CC}$)~\cite{GPW17}. 
We do not know how to lift $\Q$ or $\Q_0$ to their analogous
communication measures; this is likely to be
significantly harder.

\subsection{Lifting theorem reductions}
In this section, we prove several lifting
theorem reductions, showing that a lifting
theorem for one measure (such as $\Q_0$)
implies a lifting theorem for another measure
(such as $\Q$). Our work (including prior work~\cite{BK16}) is the first instance we
know of where such reductions are shown;
it is perhaps surprising that these reductions
can be proven without proving the lifting
theorems themselves.

\begin{theorem}\label{thm:Q0_lift}
If there is a lifting theorem for $\Q_0$
with gadget $G$, then there is also
a lifting theorem for $\Q$ with the same gadget $G$.
\end{theorem}

\begin{proof}
Fix a partial function $f$. We wish to show
that $\Q^{\CC}(f\circ G)=\tOmega(\Q(f))$
using a lifting theorem for $\Q_0$.

Let $g=\UInd_k\circ_k f$, with $k=\Theta(\Q(f))$.
By \thm{Ind}, we have $\Q(g)=\tOmega(\Q(f))$.
Next, apply the lifting theorem to $g$ to get
\begin{equation}
\Q_0^{\cc}(g\circ G)=\tOmega(\Q_0(g))=\tOmega(\Q(g))=\tOmega(\Q(f)).
\end{equation}
To complete the argument, it remains to show that
$\Q_0^{\CC}(g\circ G))=\tO(\Q^{\CC}(f\circ G))$.
Note that $g\circ G=\UInd_k\circ_k f\circ G$.
If we have a communication protocol for $f\circ G$,
we can simulate it $k$ times (and use error reduction) to obtain the correct index
with constant error. We can then use the promise of $\UInd$ to check
if the index is correct, by verifying that the second bit of the
pair at that index is $1$. This turns the algorithm into a zero-error
algorithm. Since $k=O(\log\Q(f))$, our algorithm uses only
$\tO(\Q^{\CC}(f\circ G))$ communication.
Thus $\Q^{\CC}(f\circ G)=\tOmega(\Q(f))$, as desired.
\end{proof}

\begin{theorem}\label{thm:QSZK_lift}
If there is a lifting theorem for $\QSZK$ with gadget $G$, then
there is also a lifting theorem for $\Q$ with the same gadget $G$.
\end{theorem}

By a lifting theorem for $\QSZK$, we mean a theorem that lifts
it to some communication complexity analogue $\QSZK^{\CC}$.
The only property we use of $\QSZK^{\CC}$ is that it lower bounds $\Q^{\CC}$.

\begin{proof}
Let $f$ be a partial function. Let $g=\Ind_k\circ_k f$,
where $k=\Theta(\log\Q(f))$. By \thm{Ind}, $\QS(g)=\tOmega(\Q(f))$.
By \thm{QSQSZK}, $\QSZK(g)=\Omega(\QS(g))=\tOmega(\Q(f))$. Then
\begin{equation}
\Q^{\CC}(g\circ G)=\Omega(\QSZK^{\CC}(g\circ G))
    =\tOmega(\QSZK(g))=\tOmega(\Q(f)).
    \end{equation}
Also, note that if we had a quantum communication protocol for $f\circ G$
we could easily convert it to a communication protocol for
$g\circ G=\Ind_k\circ_k f\circ G$. Thus
$\Q^{\CC}(f\circ G)=\tOmega(\Q^{\CC}(g\circ G)=\tOmega(\Q(f))$,
as desired.
\end{proof}

\begin{theorem}\label{thm:QS_lift}
If there is a lifting theorem that lifts $\QS\to\Q^{\CC}$ with gadget $G$, then
there is also a lifting theorem for $\Q$ with the same gadget $G$.
\end{theorem}

By a lifting theorem for $\QS\to\Q^{\CC}$, we mean a theorem that shows
$\Q^{\CC}(f\circ G)=\tOmega(\QS(f))$ for all partial functions $f$.
This is formally easier to prove than a $\tOmega(\Q(f))$ lower bound,
but we show it is actually equivalent.

\begin{proof}
Let $f$ be a partial function. Let $g=\Ind_k\circ_k f$,
where $k=\Theta(\log\Q(f))$. By \thm{Ind}, $\QS(g)=\tOmega(\Q(f))$. Then
\begin{equation}\Q^{\CC}(g\circ G)=\tOmega(\QS(g))=\tOmega(\Q(f)).\end{equation}
Also, note that if we had a quantum communication protocol for $f\circ G$
we could easily convert it to a communication protocol for
$g\circ G=\Ind_k\circ_k f\circ G$. Thus
$\Q^{\CC}(f\circ G)=\tOmega(\Q^{\CC}(g\circ G)=\tOmega(\Q(f))$,
as desired.
\end{proof}

In summary, what we have shown is that a lifting theorem for $\Q$ is implied by a lifting theorem for either $\Q_0$, $\QSZK$, or a $\QS\to\Q^\cc$ lifting theorem. In fact, each of these statements also has a classical analogue which remains true. Proving a lifting theorem for $\R_0$, $\SZK$, or $\RS \to \R^\cc$ would imply a lifting theorem for $\R$. 
This can be proved analogously; the only property we need is that
$\RS(\UInd_k\circ_k f)=\tOmega(\R(f))$ when $k$ is at least
polylogarithmic in $\R(f)$. An equivalent statement to this
was proven in \cite{BK16}. However, since lifting theorems for $\R$ and $\R_0$ are already known (with an index gadget \cite{GPW17}), this reduction is less interesting in the classical case, though it might still be relevant
for proving lifting theorems with other gadgets.

\section*{Acknowledgements}
We thank Scott Aaronson, Mika G\"o\"os, John Watrous, and Ronald de Wolf for helpful conversations about this work.

Most of this work was performed while the first author was at the Massachusetts Institute of Technology and the University of Maryland and the second author was at the Massachusetts Institute of Technology.
This work was partially supported by NSF grant CCF-1629809.

\appendix
\section{Proof of the hybrid argument}
\label{app:hybrid}
In this section, we prove \lem{hybrid}, restated below for convenience.

\hybrid*

\begin{proof}
We start by fixing some notation. Let the quantum query algorithm $Q$  act on $m$ qubits, 
initialized in the all-zeros state $|0^m\>$. 
A $T$-query algorithm is specified by $T+1$ unitaries $U_0, U_1, \ldots, U_T$ acting on $m$ qubits.
For any input $x\in\B^n$, the oracle $O_x$ acts as $O_x|i,b\> = |i,b\oplus x_i\>$ for all $i\in[n]$ and $b\in \B$.
The output state produced by this quantum algorithm (before measurement) on input $x$ is 
\begin{equation}|\psi_x\> = U_T O_x U_{T-1} O_x \cdots O_x U_1 O_x U_0 |0^m\>,\end{equation}
where $O_x$ is implicitly $O_x\otimes \id$ if $O_x$ acts on fewer than $m$ qubits.
Within the $m$ qubits, we further group the qubits into three registers, the first register holds an index $|i\>$, for $i\in [n]$, the second holds a qubit $|b\>$, for $b\in\B$, and the third register contains all the remaining qubits.

For a quantum algorithm outputting a Boolean function, we assume that the first qubit of $|\psi_x\>$ is measured at the end to determine the output. 
A quantum distinguishing algorithm may trace out some qubits of $|\psi_x\>$ before producing an output or it may simply output the state $|\psi_x\>$ without loss of generality, since tracing out qubits cannot increase the distance between a pair of states.

In our case we have an algorithm $Q$ that accepts $x$ and rejects $x^B$ with high probability.
To be more concrete, let us assume $Q$ has error probability $\epsilon$.
As we saw in \prop{QSQ}, such an algorithm can be made to output a mixed state $\rho_x$ such that
$\tr{\rho_x-\rho_{x^B}}\ge 1-2\epsilon$.
Since trace distance is non-increasing under partial trace~\cite[Th.~9.2]{NC00}, we get that the pure output states must also be far, and hence
$\tr{|\psi_x\rangle\langle\psi_x|-|\psi_{x^B}\rangle\langle\psi_{x^B}|}\ge 1-2\epsilon$.
This is all we need to assume about the output of the algorithm on these inputs.

We now consider the intermediate states produced by this quantum algorithm after $t$ queries to input $x$.
Let
\begin{equation}
|\psi^0_x\rangle:=U_0|0^m\rangle \quad \textrm{ and }\quad  |\psi_x^t\rangle:=U_t O_x |\psi^{t-1}_x\>.
\end{equation}
for $t\in[T]$. The final state of the algorithm is $|\psi_x^T\rangle = |\psi_x\>$, and hence we have
\begin{equation}
\tr{|\psi^T_x\rangle\langle\psi^T_x|-|\psi^T_{x^B}\rangle\langle\psi^T_{x^B}|}\ge 1-2\epsilon.
\end{equation}
We know that the states are far apart in trace distance, but we also want to bound their closeness in $\ell_2$ distance. By \eq{tr}, we have
\begin{equation}\label{eq:epsupper}
|\langle\psi^T_x|\psi^T_{x^B}\rangle|
\le\sqrt{1-(1-2\epsilon)^2} = 2\sqrt{\epsilon(1-\epsilon)}
\le 1-(1/2)(1-2\epsilon)^2.
\end{equation}
Then we have
\begin{equation}\||\psi^T_x\rangle-|\psi^T_{x^B}\rangle\|^2
=2-\langle \psi^T_{x^B}|\psi^T_x\rangle
    -\langle \psi^T_x|\psi^T_{x^B}\rangle
=2-2\re(\langle \psi^T_{x^B}|\psi^T_x\rangle)
\ge 2-2|\langle \psi^T_{x^B}|\psi^T_x\rangle|
\ge (1-2\epsilon)^2,\end{equation}
and so $\||\psi^T_x\rangle-|\psi^T_{x^B}\rangle\|\ge 1-2\epsilon$.

Hence the final states of the algorithm are far in apart in $\ell_2$ distance 
on inputs $x$ and $x^B$. We also know that the initial states $|\psi^0_x\rangle$
and $|\psi^0_{x^B}\rangle$ are identical. We keep track of how much
this distance $d_t:=\||\psi^t_x\rangle-|\psi^t_{x^B}\rangle\|$
changes for $t\in\{0,1,\dots, T\}$. For each $t$, we have
\begin{equation}
d_{t+1}=\||\psi^{t+1}_x\rangle-|\psi^{t+1}_{x^B}\rangle\|
=\|U_{t+1}O_x|\psi^t_x\rangle-U_{t+1}O_{x^B}|\psi^t_{x^B}\rangle\|
=\|O_x|\psi^t_x\rangle-O_{x^B}|\psi^t_{x^B}\rangle\|,
\end{equation}
since $U_{t+1}$ is a unitary and preserves norms. This equals
\begin{equation}\|O_{x^B}|\psi^t_x\rangle-O_{x^B}|\psi^t_{x^B}\rangle+(O_x-O_{x^B})|\psi^t_{x}\rangle\|
\le\|O_{x^B}|\psi^t_{x}\rangle-O_{x^B}|\psi^t_{x^B}\rangle\|+\|(O_x-O_{x^B})|\psi^t_{x}\rangle\|\end{equation}
\begin{equation}=d_t+\|(O_x-O_{x^B})|\psi^t_{x}\rangle\|.\end{equation}
Next, decompose $|\psi^t_{x}\rangle$ by the value of the query register.
On basis vectors when the query register is not in $B$, the unitaries
$O_x$ and $O_{x^B}$ behave the same; such vectors therefore
get mapped to zero. If $|\psi_t^{x,B}\rangle$ denotes the component
of $|\psi^t_{x}\rangle$ whose query register is in $B$, we get
\begin{equation}\|(O_x-O_{x^B})|\psi^t_{x}\rangle\|
=\|(O_x-O_{x^B})|\psi_t^{x,B}\rangle\|
\le \|O_x|\psi_t^{x,B}\rangle\|+\|O_{x^B}|\psi_t^{x,B}\rangle\|
=2\||\psi_t^{x,B}\rangle\|\end{equation}
\begin{equation}=2\cdot\sqrt{\sum_{i\in B} m_i^{t+1}},\end{equation}
where the last equality follows from the definition of 
$m_i^{t+1}$, which is defined to be the probability that the algorithm
is found to be querying position $i$ right before making query $t+1$.
The increase from $d_t$ to $d_{t+1}$ is therefore upper bounded
by $2\sqrt{\sum_{i\in B} m_i^{t+1}}$, so we have
\begin{equation}2\sum_{t=1}^T\sqrt{\sum_{i\in B} m_i^t}
\ge d_T-d_0\ge 1-2\epsilon.\end{equation}

Using the Cauchy–-Schwarz inequality on the outer sum gives
\begin{equation}2\sqrt{T}\sqrt{\sum_{t=1}^T\sum_{i\in B} m_i^t}\ge 1-2\epsilon,\end{equation}
or
\begin{equation}\sum_{t=1}^T\sum_{i\in B} m_i^t\ge\frac{(1-2\epsilon)^2}{4T} = \Omega\left(\frac{1}{T}\right),
\end{equation}
when $\epsilon$ is a constant.\footnote{This can be slightly improved
to $(1-2\sqrt{\epsilon(1-\epsilon)})/2T$ by not using the approximation in \eq{epsupper}.}
\end{proof}

\bibliographystyle{alphaurl}
\phantomsection\addcontentsline{toc}{section}{References} 
\renewcommand{\UrlFont}{\ttfamily\small}
\newcommand{\eprint}[1]{\small \upshape \tt \href{http://arxiv.org/abs/#1}{#1}}
\let\oldpath\path
\renewcommand{\path}[1]{\small\oldpath{#1}}
\bibliography{sabotage}

\end{document}